\documentclass[10pt, conference, compsocconf]{IEEEtran}
\makeatletter
\def\blfootnote{\xdef\@thefnmark{}\@footnotetext}
\makeatother

\usepackage{cite}
\usepackage{url}

\usepackage{graphicx}
\usepackage{algorithm}
\usepackage[noend]{algorithmic}
\usepackage{subfig}
\usepackage{amssymb, amsmath,graphicx,charter, latexsym}
\usepackage{enumerate}

\newtheorem{definition}{Definition}
\newtheorem{lemma}{Lemma}

\newtheorem{theorem}{Theorem}

\def\baselinestretch{1.0}
\begin{document}

\title{Providing End-to-End Delay Guarantees for Multi-hop Wireless Sensor Networks over Unreliable Channels}

\author{\IEEEauthorblockN{I-Hong Hou}
\IEEEauthorblockA{Computer Engineering and System Group \& Department of ECE\\
Texas A\&M University\\
College Station, TX, USA\\
ihou@tamu.edu}
}

\maketitle\blfootnote{
This material is based upon work partially supported by NSF under Contracts CNS-1035378, CCF-0939370, CNS-1035340, and CNS-0905397, USARO under Contract Nos. W911NF-08-1-0238 and W-911-NF-0710287, and AFOSR under Contract FA9550-09-0121.}

\begin{abstract}
Wireless sensor networks have been increasingly used for real-time surveillance over large areas. In such applications, it is important to support end-to-end delay constraints for packet deliveries even when the corresponding flows require multi-hop transmissions. In addition to delay constraints, each flow of real-time surveillance may require some guarantees on throughput of packets that meet the delay constraints. Further, as wireless sensor networks are usually deployed in challenging environments, it is important to specifically consider the effects of unreliable wireless transmissions.

In this paper, we study the problem of providing end-to-end delay guarantees for multi-hop wireless networks. We propose a model that jointly considers the end-to-end delay constraints and throughput requirements of flows, the need for multi-hop transmissions, and the unreliable nature of wireless transmissions. We develop a framework for designing feasibility-optimal policies. We then demonstrate the utility of this framework by considering two types of systems: one where sensors are equipped with full-duplex radios, and the other where sensors are equipped with half-duplex radios. When sensors are equipped with full-duplex radios, we propose an online distributed scheduling policy and proves the policy is feasibility-optimal. We also provide a heuristic for systems where sensors are equipped with half-duplex radios. We show that this heuristic is still feasibility-optimal for some topologies.

\end{abstract}


\section{Introduction}

The advance of wireless sensor networks provides an appealing solution for real-time surveillance. In real-time surveillance, wireless sensors generate flows of surveillance data and deliver them to a sink, which makes control decisions based on the data. Examples of such applications have been demonstrated in many previous work, such as \cite{IS07,SO06, XZ11}.

A major challenge for real-time surveillance is to provide end-to-end delay guarantees for packet deliveries. Designing scheduling policies that provide end-to-end delay guarantees is difficult due to two reasons. As wireless sensor networks may be deployed over a large area, some flows may require multi-hop transmissions to reach the sink. Further, wireless sensor networks are usually deployed in challenging environments, such as battlefields, forests, or underwater. Within these environment, it may be impossible to ensure that all wireless transmissions can be successfully received. Thus, a desirable policy needs to explicitly address the unreliable nature of wireless transmissions.

In this paper, we aim to address the above difficulties. We measure the performance of each surveillance flow by its \emph{timely-throughput}, defined as the throughput of packets that are delivered to the sink on time. We then propose a model that characterizes the hard per-packet end-to-end delay constraints and timely-throughput requirements of flows, the routing protocol for multi-hop transmissions, and the unreliable wireless channels. This model also considers both scenarios where sensors are equipped with full-duplex radios and half-duplex ones.

Based on the model, we establish a general framework for designing scheduling policies. We prove a sufficient condition for a scheduling policy to be \emph{feasibility-optimal}, that is, to be able to fulfill all timely-throughput requirements as long as they are feasible. We show that, based on this condition, there is a dynamic programming approach for designing policies for various types of systems.

We then consider designing online, tractable, and distributed scheduling policies. We propose a policy for systems where sensors are equipped with full-duplex radios. We prove that the proposed policy is feasibility-optimal. We also propose a simple heuristic for systems where sensors are equipped with half-duplex radios, and show that it is feasibility-optimal among certain topologies.

In addition to theoretical studies, we also provide simulation results. We compare our proposed policies against other policies. Simulation results show that our proposed policies achieve significantly better performance than others.

The rest of the paper is organized as follows. Section \ref{section:related} summarizes existing work on providing end-to-end delay guarantees. Section \ref{section:model} formally introduces our analytical model. Section \ref{section:framework} establishes a framework for designing feasibility-optimal policies. Based on the framework, Section \ref{section:sensor} proposes a feasibility-optimal policy for systems where sensors are equipped with full-duplex radios. Section \ref{section:half} proposes a heuristic for systems where sensors are equipped with half-duplex radios, and proves that the heuristic is feasibility-optimal for some topologies. Section \ref{section:simulation} demonstrates our simulation results. Finally, Section \ref{section:conclusion} concludes this paper.

\section{Related Work} \label{section:related}

Providing end-to-end delay guarantees have been an important research topic for various systems. Jayachandran and Abdelzaher \cite{PJ08} have studied this problem for distributed real-time systems where a job needs to traverse a number of processors before it is completed, and have provided a worst-case analysis for end-to-end delays. Hong, Chantem, and Hu \cite{SH11} have considered a similar problem and approached it by assigning local deadlines for each processor. Li and Eryilmaz \cite{RL11} have proposed a scheduling policy that aims to meet per-packet delay bounds and timely-throughput requirements of flows in wireline networks. However, no performance guarantees were provided for their scheduling policy. Rodoplu et al \cite{VR10} have studied the problem of estimating end-to-end delay over multi-hop wireless networks. Li et al \cite{HL09} have proposed using expected end-to-end delay for selecting path in wireless mesh networks. The expected end-to-end delay takes both queuing delay and delay caused by unsuccessful wireless transmissions. However, their work only aims at minimizing the average end-to-end delays, and cannot provide guarantees on per-packet delays. Jayachandran and Andrews \cite{PJ10} have applied a coordinated EDF scheduler for wireless networks and obtained asymptotic bounds on end-to-end delays. Li et al \cite{HL11} have used network calculus to analyze and derive upper-bound for end-to-end delays. Li, Li, and Mohapatra \cite{JL09} have proposed a distributed policy for scheduling packets with end-to-end delay guarantees. However, their work lacks theoretical guarantees on performance.

There has also been a lot of work that considers end-to-end delay guarantees for wireless sensor networks. Jiang, Ravindran, and Cho \cite{BJ09} have studied the real-time capacity of wireless sensor networks. They approach this problem by decomposing end-to-end delays into per-hop delays, and then study the probability for meeting each per-hop delay independently. Wang et al \cite{XW10} have used a similar decomposition approach and studied the problem of energy saving while providing end-to-end delay guarantees. Such decomposition approach inevitably leads to suboptimal solutions. Chipara et al \cite{OC11} have proposed a protocol for scheduling real-time flows by taking interference among sensors into account. Wang et al \cite{YW09} have investigated the distribution of end-to-end delay in wireless sensor networks. Wang et al \cite{QW11} have formulated the problem of providing end-to-end delay guarantees as an optimization problem, and have proposed a heuristic for obtaining sub-optimal solutions. Li, Shenoy, and Ramamritham \cite{HL05} have aimed at providing end-to-end delay guarantees by exploiting spatial reuse.

\section{System Model}  \label{section:model}

In this section, we present our model for multi-hop wireless sensor networks with end-to-end delay constraints. Our model extends a model proposed in \cite{IH09}, which only considers the delay constraints of packets and unreliability of wireless transmissions in a one-hop scenario.

Consider a sensor network with a set $N$ of wireless sensors. One of the sensors play the role of the sink. Sensors may generate surveillance data that need to be delivered to the sink in a timely manner, and they may relay data that are generated by other sensors. We assume that a routing tree has been constructed by some routing protocol for the sensor network. There has been a lot of work on constructing routing trees for wireless sensor networks, and \cite{AK04} provides a survey of these routing protocols. In the routing tree, the sink is the root, and hence we use $r$ to represent the sink. When a sensor $n$ has a data packet, either one generated by itself or one that is forwarded to it from other sensors, it may forward the data to its parent, denoted by $h(n)$, in the routing tree. Figure \ref{fig:model:example} shows an example of such a sensor network. A data packet is said to be \emph{delivered} if it reaches the sink. A sensor may generate multiple flows of data. For example, one sensor may generate data on both temperature and humidity. We denote the set of flows in the wireless sensor network by $F$, and $n(f)$ as the sensor that generates data of flow $f$.

\begin{figure}[t]
\includegraphics[height = 3.2in, angle = 90]{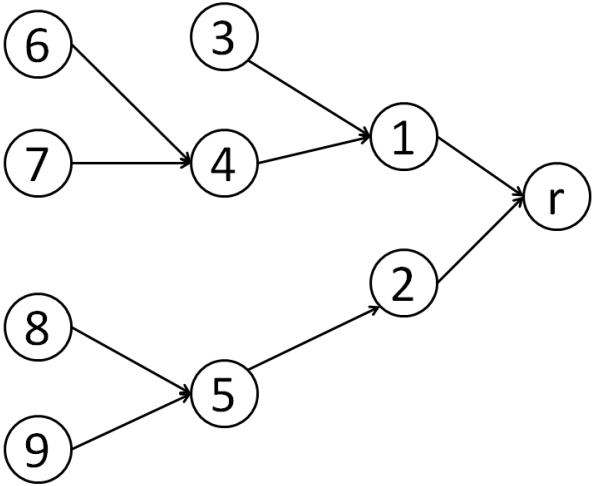}
\caption{An example of the system consists of 10 sensors. In the example, we have $h(1)=h(2) = r$, $h(3) = h(4)=1$, $h(5)=2$, etc.} \label{fig:model:example}
\end{figure}

We assume that time is slotted and numbered by $t=1,2,\dots$. The length of a time slot is set to be the time needed for a sensor to transmit one data packet. Time slots are further grouped into \emph{intervals}, where each interval consists of $T$ consecutive time slots in $(kT,(k+1)T]$, for some $k$. At the beginning of each interval, each flow in $F$ obtains some surveillance data and generates a data packet. We say that a the data packet of flow $f$ is generated at the $\tau_f^{th}$ time slot in each interval, so as to account for the latency caused by sensing and data processing. We assume that all data packets are delay-constrained, and data packets generated in one interval need to be delivered to the sink before the end of the interval. If a data packet is not delivered before the end of the interval, the packet is no longer useful for the sink. In this case, we say that the packet \emph{expires}, and drop the packet from the system. Thus, we can guarantee that all data received by the sink have delays no larger than $T$ time slots.

We consider both cases when sensors are equipped with full-duplex radios and when they are equipped with half-duplex radios. When sensors are equipped with full-duplex radios, they can transmit and receive data packets simultaneously. We also assume that the transmissions of different sensors do not interfere with each other by, for example, allocating different sensors on different subcarriers in an orthogonal frequency-division multiple access (OFDMA) system. A system where sensors are equipped with full-duplex radios is called a \emph{full-duplex system}.

The assumptions made for full-duplex systems may exceed current hardware limitations of wireless sensor network. Thus, we also consider systems where sensors are equipped with half-duplex radios. In such systems, sensors cannot transmit and receive data packets simultaneously. That is, when a sensor $n$ transmits, its parent $h(n)$ cannot transmit, or the transmission by $n$ encounters a collision and the transmission fails. Moreover, we assume that a sensor can receive at most one transmission in a time slot. That is, if we have sensors $n$ and $m$ with $h(n)=h(m)$, then at most one of them can transmit in a time slot. Finally, we assume that different transmissions do not interfere with each other except the two cases discussed above. This can be done by, for example, scheduling transmissions that may interfere with each other in different channels. A system where radios are equipped with half-duplex radios is called a \emph{half-duplex system}.

We consider the unreliable nature of wireless transmissions. To be more specific, we say that when a sensor $n$ transmits a data packet to its parent, $h(n)$, $h(n)$ correctly receives the packet with probability $p_n$. We also assume that, by implementing ACKs, the sensor $n$ has feedback information on whether its transmission is correctly received by $h(n)$, and it can retransmit the same packet in the case that a previous transmission fails.

As wireless transmissions are unreliable, it may be impossible to deliver all data packets to the sink on time. Instead, each flow $f$ requires a portion $q_f$ of packets to be delivered on time. That is, let $e_f(k)$ be the indicator function that the data packet of flow $f$ in the $k^{th}$ interval is delivered to the sink on time.  Each flow $f$ then requires that, with probability one, \[\liminf_{K\rightarrow\infty}\frac{\sum_{k=1}^Ke_f(k)}{K}\geq q_f.\] We call $\frac{\sum_{k=1}^Ke_f(k)}{K}$ as the \emph{timely-throughput} of flow $f$ up to interval $K$, and $q_f$ as the \emph{timely-throughput requirement} on flow $f$.

In this paper, we aim to design scheduling policies that \emph{fulfills} timely-throughput requirements of all flows as long as they are \emph{strictly feasible}. These terms are formally defined as follows:

\begin{definition}
A system is said to be \emph{fulfilled} by some scheduling policy if, under this policy, $\liminf_{K\rightarrow\infty}\frac{\sum_{k=1}^Ke_f(k)}{K}\geq q_f$ with probability one, for all flow $f\in F$.
\end{definition}

\begin{definition}
A system, either a full-duplex system or a half-duplex one, is \emph{feasible} if there exists some scheduling policy that fulfills it.
\end{definition}

\begin{definition}
A system, either a full-duplex system or a half-duplex one, is \emph{strictly feasible} if $q_f>0$ for all flows $f$, and there exists some $\epsilon>0$ such that the system is still feasible when each flow $f$ requires a timely-throughput of $(1+\epsilon)q_f$.
\end{definition}

\begin{definition}
A scheduling policy is \emph{feasibility-optimal} for full-duplex system, or half-duplex system, if it fulfills all strictly feasible full-duplex systems, or all strictly feasible half-duplex systems, respectively.
\end{definition}

We limit our discussions on strictly feasible systems only to simplify the proof of Theorem \ref{theorem:framework:sufficient} in the following section. As $\epsilon$ can be arbitrarily small, this limitation is not restrictive.

\section{A Framework for Scheduling Policies}   \label{section:framework}

In this section, we describe a sufficient condition for a policy to be feasibility-optimal. We then show a dynamic program approach that derives feasibility-optimal policies by employing this sufficient condition. This condition is based on the concept of \emph{debt}:

\begin{definition}
The \emph{debt} of $f$ in the $K^{th}$ interval is defined as $d_f(K):=Kq_f-\sum_{k=1}^Ke_f(k)$, where $e_f(k)$ is the indicator function that the packet for $f$ in interval $k$ is delivered on time.
\end{definition}

It is easy to show that a system is fulfilled under some scheduling policy if and only if $\limsup_{K\rightarrow\infty}\frac{d_f(K)}{K}\leq 0$. Based on the concept of debt, we establish a sufficient condition for a policy to be feasibility-optimal for full-duplex systems or half-duplex systems. The condition is similar to the one introduced in \cite{IH10}, which only considers one-hop transmissions, and is based on the following theorem:

\begin{lemma}[Telescoping Lemma]
\label{theorem:framework:lyapunov} Let $L(k)$ be a non-negative Lyapunov function depending only on $\mathcal{F}_k$, which denotes the set of all events in the system in the first $k$ intervals, i.e., $L(k)$ is adapted to $\mathcal{F}_k$. Suppose there exist positive constants $B>0$, $\delta>0$, and a stochastic process $f(k)$ also adapted to $\mathcal{F}_k$, such that:
\begin{equation}
E[L(k+1)-L(k)|\mathcal{F}_k]\leq B-\delta E[f(k)|\mathcal{F}_k],
\end{equation}
then
$\limsup_{K\rightarrow\infty}\frac{1}{K}\sum_{k=0}^{K-1}E[f(k)]\leq B/\delta$, where $E[x]$ is the expected value of $x$.
\end{lemma}

\begin{theorem}
\label{theorem:framework:sufficient}
A scheduling policy is feasibility-optimal for full-duplex system, or half-duplex system, if, given $d_f(k)$, it maximizes \[\sum_{f\in F}d_f(k)^+E[e_f(k+1)]\] in the $(k+1)^{th}$ interval, for all $k$, where $x^+:=\max\{x,0\}$, for all full-duplex systems, or half-duplex systems, respectively.
\end{theorem}
\begin{proof}
Consider a strictly feasible system where flow $f$ requires a timely-throughput of $q_f >0$. There exists some $\epsilon>0$ such that this system is fulfilled by some stationary randomized scheduling policy, $\eta^0$, when each flow $f$ requires a timely-throughput of $(1+\epsilon)q_f$. Thus, under $\eta^0$, we have $E[e_f(k)]\geq (1+\epsilon)q_f$.

Define a Lyapunov function $L(k):=\frac{1}{2}\sum_f(d_f(k)^+)^2$. Since $d_f(k+1)=d_f(k)+q_f-e_f(k+1)$, we have
\begin{align*}
&L(k+1)=\frac{1}{2}\sum_{f\in F}(d_f(k+1)^+)^2\\
\leq& \frac{1}{2}\sum_{f\in F}(d_f(k)^++q_f-e_f(k+1))^2\\
=&\frac{1}{2}\sum_{f\in F}(d_f(k)^+)^2+\sum_{f\in F}d_f(k)^+(q_f-e_f(k+1))\\
&+\sum_{f\in F}(q_f-e_f(k+1))^2\\
\leq& L(k)+\sum_{f\in F}d_f(k)^+(q_f-e_f(k+1))+B,\tag{2} \label{equation:framework:lyapunov}
\end{align*}
for some constant $B$, as $q_f$ and $e_f(k+1)$ are both bounded by 1, for all $f$.

As $E[e_f(k+1)]\geq (1+\epsilon)q_f$ under $\eta^0$, for all $f$, we have, under $\eta^0$,
\begin{align*}
&E[\sum_{f\in F}d_f(k)^+(q_f-e_f(k+1))|\mathcal{F}_k]\\
=&\sum_{f\in F}d_f(k)^+(q_f-E[e_f(k+1)])\\
=&\sum_{f\in F}d_f(k)^+q_f-\sum_{f\in F}d_f(k)^+E[e_f(k+1)]\\
\leq &-\epsilon\sum_{f\in F}d_f(k)^+q_f.\tag{3}
\end{align*}

Let $\eta^{max}$ be the policy that maximizes $\sum_{f\in F}d_f(k)^+E[e_f(k+1)]$. Under $\eta^{max}$, we also have
\[
E[\sum_{f\in F}d_f(k)^+(q_f-e_f(k+1))|\mathcal{F}_k]\leq-\epsilon\sum_{f\in F}d_f(k)^+q_f.\]
Thus, by Eq. (\ref{equation:framework:lyapunov}), we have, under $\eta^{max}$,
\begin{align*}
&E[L(k+1)-L(k)|\mathcal{F}_k]\\
\leq& E[\sum_{f\in F}d_f(k)^+(q_f-e_f(k+1))|\mathcal{F}_k]+B\\
\leq &B-\epsilon\sum_{f\in F}d_f(k)^+q_f.
\end{align*}
Thus, by Lemma \ref{theorem:framework:lyapunov}, we have
\[\limsup_{K\rightarrow\infty}\frac{1}{K}\sum_{k=1}^K\sum_{f\in F}E[d_f(k)^+q_f]\leq B/\epsilon.\tag{4}\label{equation:framework:sufficient}\]
Finally, Lemma 4 in \cite{IH10} shows that Eq. (\ref{equation:framework:sufficient}) implies that $\limsup_{K\rightarrow\infty}(\frac{d_f(K)^+}{K})=0$, for all $f$, and this system is fulfilled by $\eta^{max}$. Thus, $\eta^{max}$ is feasibility-optimal.
\end{proof}

\subsection{A Dynamic Programming Approach for Scheduling Policies}

We now introduce a dynamic programming approach for designing scheduling policies that aims at maximizing $\sum_{f\in F}d_f(k)^+E[e_f(k+1)]$ in the $(k+1)^{th}$ interval by making scheduling decisions for each time slot within the interval. We first describe the system evolution within an interval by a Markov decision process. In the $t^{th}$ time slot within the interval, we represent the state of the system by $t$ and the position of the packet of each flow. We denote the position of the packet of flow $f$ by $c_f(t)$, where $c_f(t)=n$ if $n$ has the packet and can transmit it in the $t^{th}$ time slot, and $c_f(t)=\phi$ if the packet of $f$ is yet to be generated, as the packet of flow $f$ will not be generated until the $\tau_f^{th}$ time slot in the interval. The evolution of $c_f(t+1)$ can then be described as follows: $c_f(t+1)=n(f)$ when $t+1=\tau_f$, where $n(f)$ is the sensor that generates the packet of $f$; $c_f(t+1)=h(c_f(t))$ with probability $p_{c_f(t)}$ if $c_f(t)$ transmits the packet of $f$ in the $t^{th}$ time slot, where $h(c_f(t))$ is the parent of $c_f(t)$ in the routing tree and $p_{c_f(t)}$ is the channel reliability between $c_f(t)$ and $h(c_f(t))$; and $c_f(t+1) = c_f(t)$, otherwise.

We say that the interval ends when $t=T+1$. Thus, the packet of $f$ is delivered, and $e_f(k+1)=1$, if $c_f(T+1)=r$. Let $V(t,\{c_f(t)\})$ be the value of $\sum_{f\in F}d_f(k)^+E[e_f(k+1)]$ when the state of the system is represented by $\{c_f(t)\}$ at the $t^{th}$ time slot under some policy. By Theorem \ref{theorem:framework:sufficient}, a policy is feasibility-optimal if it maximizes $V(1, \{c_f(1)\})$, where $\{c_f(1)\}$ is the state of the system at the beginning of the interval. Let $V^{max}(t,\{c_f(t)\})$ be the maximum value of $V(t,\{c_f(t)\})$. We then have the recursive relation: $V^{max}(t,\{c_f(t)\})=\max_{a\in A(t,\{c_f(t)\})}E[V^{max}(t+1,\{c_f(t+1)\})]$, where $A(t,\{c_f(t)\})$ is the set of possible schedule decisions under the current state. Hence $V^{max}(t,\{c_f(t)\})$ can be obtained by dynamic programming. Moreover, a policy that makes its schedule decisions by choosing $\arg\max_{a\in A(t,\{c_f(t)\})}E[V^{max}(t+1,\{c_f(t+1)\})]$ for all states is feasibility-optimal.

While we can derive feasibility-optimal policies from the above approach, this approach may require high computation overhead, as the number of states for the system is as large as $(T+1)(|N|+1)^{|F|}$. In the following sections, we demonstrate that there is a simple online policy for full-duplex systems. We also introduce a heuristic for half-duplex system and show that the heuristic is feasibility-optimal under some restrictions.

\section{An Online Scheduling Policy for Full-Duplex Systems}
\label{section:sensor}

We now introduce an online scheduling policy for full-duplex systems. The policy is very simple: in each time slot, a sensor $n$ picks the flow that has the largest debt among those that it currently holds their packets, and transmits its packet. In other words, in each time slot $t$ of the $(k+1)^{th}$ interval, $n$ schedules the packet from $\arg\max_{f:c_f(t)=n}d_f(k)$. We call this policy the \emph{Greedy Forwarder}. In addition to low complexity, this policy is also distributed and does not require a centralized scheduler. Moreover, as we establish below, this simple policy is indeed feasibility-optimal.

\begin{theorem} \label{theorem:sensor:optimal}
The Greedy Forwarder is feasibility-optimal.
\end{theorem}
\begin{proof}
By Theorem \ref{theorem:framework:sufficient}, we can show that the Greedy Forwarder is feasibility-optimal by showing that it maximizes $\sum_{f\in F}d_f(k)^+E[e_f(k+1)]$. To show this, we prove the following two claims by induction on the size of the network,$|N|$:

\begin{enumerate}
\item The Greedy Forwarder maximizes $\sum_{f\in F}d_f(k)^+E[e_f(k+1)]$.
\item Suppose a sensor generates packets from flow $f_1,f_2,\dots$, where $d_{f_1}(k)\geq d_{f_2}(k)\geq\dots$. Also assume that this sensor generates packets at time $t_1\leq t_2\leq\dots$, and the sensor has control over which flow among $\{f_1,f_2,\dots\}$ is to generate packets at time $t_1,t_2,\dots$, respectively. Then, with all other conditions fixed, by selecting $f_1$ to generate its packet at times $t_1$, $f_2$ to generate its packet at time $t_2$, etc, $\sum_{f\in F}d_f(k)^+E[e_f(k+1)]$ for the whole system is maximized.
\end{enumerate}

We first discuss the case when $|N|=1$, in which the sink $r$ is the only sensor in the system. As there is only one sensor in the system, there are no scheduling decisions to be made, and hence Claim (1) holds. Moreover, a packet of flow $f$ is delivered, and hence $e_f(k+1)=1$, if it is generated before the $T^{th}$ time slot. Thus, Claim (2) holds as $\sum_{f\in F}d_f(k)^+E[e_f(k+1)]$ is maximized by generating packets in decrement order of their debts.

Assume that both Claim (1) and Claim (2) hold for all networks with size $|N|=M$. We now show that these two claims also hold for all networks with $|N|=M+1$. We pick a leaf node, denoted by $n_0$, in the routing tree. We assume that $n_0$ generates packets for flows $f_1,f_2,\dots$, with $d_{f_1}(k)\geq d_{f_2}(k)\geq\dots$, at times $t_1\leq t_2\leq\dots$, and $n_0$ has control over which flows among $\{f_1,f_2,\dots\}$ is to generate packets at times $t_1,t_2,\dots$, respectively. We also assume that $n_0$ uses a work-conserving policy, i.e., it always schedules a transmission as long as it holds a packet\footnote{Obviously, a policy cannot lose its optimality by making more transmissions. Thus, this assumption is not restrictive.}. Under this policy, $n_0$ successfully transmits packets at times $\hat{t}_1\leq \hat{t}_2\leq\dots$. We note that $\hat{t}_1,\hat{t}_2,\dots$ are random variables whose distributions are determined by the channel reliability between $n_0$ and $h(n_0)$. Moreover, we have $\hat{t}_i\geq t_i$, for all $i$, as it is impossible to successfully transmit $i$ packets before at least the same amount of packets are generated.  Note that $\hat{t}_1,\hat{t}_2,\dots$ are not influenced by the order of packet generations and scheduling decisions, as each transmission made by $n_0$ is successful with probability $p_{n_0}$, regardless which packet is being transmitted.

Given $\hat{t}_1,\hat{t}_2,\dots$,  $n_0$ can effectively determine which packets are to be successfully transmitted at times $\hat{t}_1, \hat{t}_2,\dots$, by choosing the order of packet generations and scheduling decisions. The only restriction for $n_0$ is that the packet from a flow $f$ cannot be transmitted before it is generated. If the packet from a flow $f$ is successfully transmitted by $n_0$ at time $t$, $h(n_0)$ receives the packet at time $t$ and can transmit the packet starting at time $t+1$. Thus, this system is equivalent to one with $n_0$ removed, making the size of the network to be $M$, and packets from flows $f_1, f_2,\dots$ are generated at sensor $h(n_0)$ at times $\hat{t}_1+1, \hat{t}_2+1,\dots$. By the induction hypothesis on this system with $M$ sensors, $\sum_{f\in F}d_f(k)^+E[e_f(k+1)]$ is maximized if the packet from flow $f_i$ is generated at time $\hat{t}_i+1$, for all $i$. As $n_0$ can make this happen by choosing flow $f_i$ to generate a packet at time $t_i$ and follow the Greedy Forwarder, both Claim (1) and Claim (2) hold when $n_0$ is included in the system with size $M+1$.

By induction on $|N|$, we have that both claims hold for all systems, and hence the Greedy Forwarder is feasibility-optimal.
\end{proof}

We close this section by discussing some implementation issues of the Greedy Forwarder. As noted above, under the Greedy Forwarder, every sensor makes scheduling decisions solely based on the packets it holds. Still, the Greedy Forwarder requires each sensor to have the knowledge of debts for all flows in the current interval. In practice, this may be impractical. In particular, for a large network, sensors that are far away from the sink can only obtain delayed information on debts of flows. However, as we will show in Section \ref{section:simulation}, when sensors apply the Greedy Forwarder with delayed information on debts, the resulting performance can still be optimal. This is because, as the net change of debt, $|d_f(k+1)-d_f(k)|$, is bounded, the difference between the delayed information on debt and the actual current debt is also bounded for all flows. Therefore, a sensor that only has delayed information on debts will make similar scheduling decisions as one that has information on the current values of debts.

\section{A Heuristic for Half-Duplex Systems}
\label{section:half}

In this section, we propose a policy for half-duplex systems and show that this policy is feasibility-optimal for some particular scenarios.

Recall that there are two important limitations for half-duplex systems. First, as a sensor cannot transmit and receive simultaneously, a sensor $n$ cannot transmit when its parent, $h(n)$, is also transmitting. Second, a sensor can only receive one packet at a time, and hence two sensors $n$ and $m$ with $h(n)=h(m)$ cannot transmit simultaneously. To address theses challenges, we propose a policy, namely, the \emph{Closest Sensor First Policy}. We first define $d_n(t):=\max_{f:c_f(t)=n}d_f(k)$. In other words, $d_n(t)$ is the largest debt of flows that $n$ holds at the $t^{th}$ time slot. If the sensor $n$ does not hold any packets, $d_n(t)$ is defined to be $-\infty$. The Closest Sensor First Policy can be described iteratively as follows: First, we examine all sensors that are one-hop away from the root $r$, that is, we examine all sensors such that $h(n)=r$. We then schedule the sensor with the largest $d_n(t)$ to transmit, and the sensor transmits the packet from the flow with the largest debt. Next, we examine sensors that are two-hop away from the root, that is, sensors with $h(h(n))=r$. If $h(n)$ is scheduled to transmit in the first step, then $n$ cannot transmit. Otherwise, a sensor $n$ is scheduled to transmit the packet from the flow with the largest debt if $d_n(t)$ is the largest among all sensors $m$ with $h(m)=h(n)$. In the above procedure, ties are broken arbitrarily, and we carry the procedure iteratively.

In summary, a sensor $n$ who is $g$-hop away from the root is scheduled in the $g^{th}$ iteration if: (i) $h(n)$ is not scheduled in the previous iteration, and (ii) $d_n(t)$ is the largest among all $m$ with $h(m)=h(n)$. Fig. \ref{fig:half:example} shows an example that illustrates the Closest Sensor First Policy. In the example, we number sensors by their respective $d_n(t)$. Thus, we have $d_n(t)=n$. In the first iteration, sensor 5 is scheduled as $d_5(t)>d_3(t)$. In the second iteration, sensor 2 is scheduled as $d_2(t)>d_1(t)$. Note that sensor 4 cannot be scheduled as its parent, sensor 5, has already been scheduled. Finally, in the third iteration, sensor 7 and sensor 9 are scheduled.

\begin{figure}[t]
\includegraphics[height = 3.0in, angle = 90]{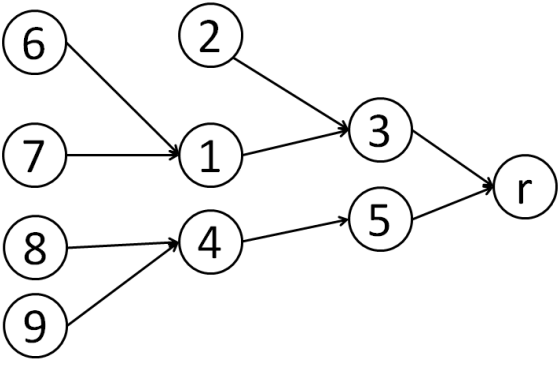}
\caption{An example that illustrates the Closest Sensor First policy.} \label{fig:half:example}
\end{figure}

Next, we show that the Closest Sensor First Policy is feasibility-optimal if all flows are generated by the same sensor $n_0$, i.e., $n(f)=n_0$, for all $f$, and each flow generates a packet at the beginning of the interval, i.e., $\tau_f=1$, for all $f$. In such a system, only sensors on the path between $n_0$ and $r$ are involved in forwarding messages. Thus, we call such a system as a \emph{path-topology system}.

\begin{theorem}
The Closest Sensor First Policy is feasibility-optimal for path-topology systems.
\end{theorem}
\begin{proof}
By Theorem \ref{theorem:framework:sufficient}, we can show that the Closest Sensor First Policy is feasibility-optimal for path-topology systems by establishing that this policy maximizes $\sum_{f\in F}d_f(k)^+E[e_f(k+1)]$ in the $(k+1)^{th}$ interval.

Let $\gamma_n[v]$ be the number of transmissions that sensor $n$ needs to make to successfully transmit $v$ packets to its parent. Note that this does not imply that sensor $n$ successfully transmits $v$ packets at the $\gamma_n[v]^{th}$ time slot in the interval, as there are time slots that sensor $n$ is not scheduled due to the constraints of half-duplex systems. We also note that $(\gamma_n[v]-\gamma_n[v-1])$ is a geometric random variable with mean $1/p_n$, as the channel reliability between $n$ and $h(n)$ is $p_n$. In practice, the values of $\{\gamma_n[v]\}$ cannot be obtained at the beginning of the interval. We will show that, even when the values of $\gamma_n[v]$ are given for all $n$ and $v$, there is no policy that can achieve larger $\sum_{f\in F}d_f(k)^+e_f(k+1)$ than the Closest Sensor First Policy, and hence than the Closest Sensor First Policy maximizes $\sum_{f\in F}d_f(k)^+E[e_f(k+1)]$.

We order flows so that $d_{f_1}(k)\geq d_{f_2}(k)\geq\dots$. Let $\eta$ and $\eta^*$ be the Closest Sensor First Policy and another policy that maximizes $\sum_{f\in F}d_f(k)^+E[e_f(k+1)]$, respectively. Further, let $\theta_{f_i}$ and $\theta^*_{f_i}$ be the times that the packet from flow $f_i$ is delivered under $\eta$ and $\eta^*$, respectively. If the packet from flow $f_i$ is not delivered on time under $\eta$, or $\eta^*$, we set $\theta_{f_i}$, or $\theta^*_{f_i}$, to be $T+1$. Thus, under $\eta$, or $\eta^*$, we have $e_{f_i}(k+1)=1(\theta_{f_i}<T+1)$, or $e_{f_i}(k+1)=1(\theta^*_{f_i}<T+1)$, respectively.

By the design of the Closest Sensor First Policy, we have $\theta_{f_1}\leq \theta_{f_2}\leq\dots$. Suppose there exists some $i$ so that $\theta^*_{f_i}>\theta^*_{f_{i+1}}$ under $\eta^*$. We can modify $\eta^*$ so that whenever it schedules the packet from flow $f_{i}$, it schedules the packet from flow $f_{i+1}$ instead, and vice versa. Under this modification, the packet of $f_i$ is delivered on the $(\theta^*_{f_{i+1}})^{th}$ time slot, and the packet of $f_{i+1}$ is delivered on the $(\theta^*_{f_{i}})^{th}$ time slot. If both $\theta^*_{f_i}$ and $\theta^*_{f_{i+1}}$ are smaller than $T+1$, both packets are still delivered on time after this modification, and hence the value of $\sum_{f\in F}d_f(k)^+e_f(k+1)$ is not influenced. If both $\theta^*_{f_i}$ and $\theta^*_{f_{i+1}}$ are larger than $T$, neither packets are delivered on time after this modification, and the value of $\sum_{f\in F}d_f(k)^+e_f(k+1)$ is not influenced. However, if $\theta^*_{f_i}>T\geq\theta^*_{f_{i+1}}$, the packet of $f_{i}$ is delivered on time and the packet of $f_{i+1}$ is not after the modification, and the value of $\sum_{f\in F}d_f(k)^+e_f(k+1)$ will not decrease, as $d_{f_i}(k)\geq d_{f_{i+1}}(k)$, with the modification. In sum, the value of $\sum_{f\in F}d_f(k)^+e_f(k+1)$ will not decrease with the modification. Thus, we can repeat this procedure until $\theta^*_{f_1}\leq\theta^*_{f_2}\leq\dots$ without decreasing the value of $\sum_{f\in F}d_f(k)^+e_f(k+1)$.

From now on, we assume that $\theta^*_{f_1}\leq\theta^*_{f_2}\leq\dots$ under $\eta^*$. We claim that, under this assumption, $\theta_{f_i}\leq \theta^*_{f_i}$ for all $i$. We prove this claim by induction on the number of flows. When there is only one flow in the system, the Closest Sensor First Policy schedules a transmission for flow 1 in every time slot, and hence $\theta_{f_1}\leq \theta^*_{f_1}$.

Assume that $\theta_{f_i}\leq \theta^*_{f_i}$ for all $i$ when the system has $I$ flows. We now consider the case when the system has $I+1$ flows. Under the Closest Sensor First Policy, whether the packet of a flow $f_i$ with $i\leq I$ is scheduled is not influenced by whether the flow $f_{I+1}$ is present in the system. Thus, the value of $\theta_{f_i}$ is the same as in the case when the system only has $I$ flows, for all $i\leq I$. We then have $\theta_{f_i}\leq \theta^*_{f_i}$ for all $i\leq I$ by the induction hypothesis. Therefore, we only need to prove that $\theta_{f_{I+1}}\leq \theta^*_{f_{I+1}}$. If $\theta^*_{f_{I+1}}=T+1$, i.e., the packet of flow $f_{I+1}$ is not delivered on time under $\eta^*$, then $\theta_{f_{I+1}}\leq \theta^*_{f_{I+1}}$ holds.

Consider the case $\theta^*_{f_{I+1}}\leq T$. Suppose that, under $\eta^*$, there is some time during the interval when the packet from flow $f_{i+1}$ is closer to the root than the packet from flow $f_i$, for some $i$. Pick $i'$ to be the smallest number so that the packet from flow $f_{i'+1}$ is closer to the root than the packet from flow $f_{i'}$ at some time $t$ during the interval. Now we can pick $t_1$ to be the largest time before $t$ such that the packet from flow $f_{i'+1}$ and that from flow $f_{i'}$ are held by the same sensor. Such $t_1$ exists as both packets are held by the sensor that generates all packets at the beginning of the interval. We then pick $t_2$ to be the smallest time after $t$ such that both packets are held by the same sensor. Such $t_2$ exists as we assume that the packet from flow $f_{i'}$ is delivered earlier than the packet from flow $f_{i'+1}$. Thus, in any time slot in $(t_1,t_2)$, the packet from flow $f_{i'+1}$ is always closer to the root than that from flow $f_{i'}$. Now, we can modify $\eta^*$ for time slots in $[t_1,t_2]$ so that when it schedules $i'$, it schedules $i'+1$ instead, and vice versa. After this modification, the packet from flow $f_{i'}$ is always closer to the root than that from flow $f_{i'+1}$ during $(t_1,t_2)$. Further, this modification does not influence $\theta_{f}$ for any flow $f$. We repeat this modification until such $i'$ does not exist. From now on, we can assume that, at any point of time, the packet from flow $f_i$ is not closer to the root than the packet from flow $f_j$ if $j<i$.

We prove that $\theta_{f_{I+1}}\leq\theta^*_{f_{I+1}}$ by contradiction. Note that after the packet of flow $f_I$ is delivered, $\eta$, or $\eta^*$, needs to schedule the packet of flow $f_{I+1}$ an addition number of $(\theta_{f_{I+1}}-\theta_{f_{I}})$, or $(\theta^*_{f_{I+1}}-\theta^*_{f_{I}})$, times before it is delivered, respectively. By the induction hypothesis, $\theta_{f_{I}}\leq\theta^*_{f_{I}}$. Therefore, if $\theta_{f_{I+1}}>\theta^*_{f_{I+1}}$, we have $\theta_{f_{I+1}}-\theta_{f_{I}}>\theta^*_{f_{I+1}}-\theta^*_{f_{I}}$, and, under $\eta^*$, there are times that the packet from flow $f_{I+1}$ is scheduled while it would not be schedule under $\eta$. We call these times the \emph{inversion times} and denote them by $\hat{t}_1<\hat{t}_2<\dots<\hat{t}_M$. We assume that, among all policies that deliver packets at times $\theta^*_{f_{1}}, \theta^*_{f_{2}}, \dots$, $\eta^*$ is one that has the smallest number of inversion times. Moreover, among those policies that have the smallest number of inversion times, $\eta^*$ is one that maximizes $\hat{t}_M$.

At time $\hat{t}_M$, the sensor $c_{f_{I+1}}(\hat{t}_M)$ holds the packet from flow $f_{I+1}$ and schedules it under $\eta^*$, while $c_{f_{I+1}}(\hat{t}_M)$ would not schedule this packet under $\eta$. There are two possibilities: first, the sensor $c_{f_{I+1}}(\hat{t}_M)$ holds another packet, and $\eta$ would schedule it; and second, under $\eta$, the sensor $c_{f_{I+1}}(\hat{t}_M)$ would not transmit, as its parent, $h(c_{f_{I+1}}(\hat{t}_M))$, would be scheduled for transmission. In the first case, under $\eta^*$, the transmission of $f_{I+1}$ cannot be successful. Otherwise, at time $\hat{t}_M+1$, the packet of $f_{I+1}$ is closer to the root than the other packet that $c_{f_{I+1}}(\hat{t}_M)$ holds, and violates our previous assumption. Thus, for this case, we can modify $\eta^*$ so that $c_{f_{I+1}}(\hat{t}_M)$ schedules the other packet, and this modification will not influence the deliver times of packets. In this case, there are no inversion times at and after time $\hat{t}_M$ after the modification. As this modification does not create new inversion times, we obtain a policy that has a smaller number of inversion times than $\eta^*$, which contradicts our assumption in the last paragraph. Now consider the second case, that is, the sensor $c_{f_{I+1}}(\hat{t}_M)$ would not be scheduled by $\eta$ because $\eta$ would schedule its parent for transmission. As $\hat{t}_M$ is the largest inversion time, there will be a time after $\hat{t}_M$ such that $h(c_{f_{I+1}}(\hat{t}_M))$ is scheduled to transmit a packet, and the packet from flow $f_{I+1}$ is not scheduled. We call this time $\hat{t}'$. At time $\hat{t}'$, either sensor $c_{f_{I+1}}(\hat{t}_M)$ or $h(c_{f_{I+1}}(\hat{t}_M))$ holds the packet from flow $f_{I+1}$. We can modify the schedule so that $h(c_{f_{I+1}}(\hat{t}_M))$ is scheduled for transmission at time $\hat{t}_M$, instead of $c_{f_{I+1}}(\hat{t}_M)$, and the packet from flow $f_{I+1}$ is scheduled for transmission at time $\hat{t}'$. Note that, by our previous assumption, the packet from flow $f_{I+1}$ is not closer to the sink than any other packets. In other words, every sensor that is farther from the sink than the one holds the packet from flow $f_{I+1}$ does not hold any packets. Therefore, this modification will not violate any interference constraints of the half-duplex system. This modification does not influence the delivery times of packets and does not increase the number of inversion times. Moreover, after applying the modification, the largest inversion time becomes $\hat{t}'$, which contradicts our assumption that $\eta^*$ maximizes the largest inversion time.

In summary, we have established that $\theta_{f_i}\leq \theta^*_{f_i}$ for all flows when there are $I+1$ flows in the system. By induction, we have $\theta_{f_i}\leq \theta^*_{f_i}$ for all flows, for all path-topology systems. Therefore, the Closest Sensor First Policy is feasibility-optimal for path-topology systems.
\end{proof}

\section{Simulation Results}    \label{section:simulation}

In this section, we present our simulation results. We adopt the simulation settings in \cite{HL05}, where each flow generates one packet every 20 ms, and it takes 2 ms for a sensor to make a transmission. Thus, we set the duration of a time slot to be 2 ms, and each interval consists of 10 time slots.

We consider the network topology as shown in Fig. \ref{fig:model:example}. For each sensor $n$, its channel reliability, $p_n$, is randomly selected within $[0.4, 0.9]$. We assume that each of sensor 3, sensor 5, sensor 6, sensor 7, sensor 8, and sensor 9 generates two flows. Therefore, there are a total number of 12 flows in the system. To better illustrate our simulation results, we assume that for each of these sensors, one flow requires a timely-throughput of $\alpha$, and the other requires a timely-throughput of $\beta$. We define the \emph{timely-throughput region} of a policy to be the region consists of all $(\alpha, \beta)$ that can be fulfilled by the policy. We can then evaluate the performance of a policy by its timely-throughput region.

For all scenarios, we conduct the simulation for 3000 intervals, i.e., one minute in the simulation environment. A system is said to be fulfilled if, by the end of the simulation, the debts are less than 90 for all flows, which means the actual timely-throughput that a flow has is at least $(q_f-0.03)$.

We consider both the full-duplex system and half-duplex system. For the full-duplex system, we compare our proposed policy, the Greedy Forwarder, against two other policies. We consider a policy where each sensor randomly chooses a packet that it holds to transmit in each time slot. The policy is called the Random policy. We also consider another policy where sensors give priorities to flows with higher timely-throughput requirements, and break ties randomly. The policy is called the Static Priority policy.

The simulation results for the full-duplex system is shown in Fig. \ref{fig:simulation:full-duplex}. The Greedy Forwarder achieves the largest timely-throughput region, as it is indeed feasibility-optimal. The performance of the Static Priority policy is close to optimal when either $\alpha$ is much larger than $\beta$, or vice versa, as it gives higher priorities to flows with larger timely-throughput requirements. On the other hand, the Static Priority policy inevitably starves flows with smaller timely-throughput requirements. Thus, it results in poor performance when $\alpha$ is close to $\beta$. The performance of the Random policy is also far from optimal, as it does not take the timely-throughput requirements of flows into account.

\begin{figure}[t]
\includegraphics[height = 3.0in, angle = 90]{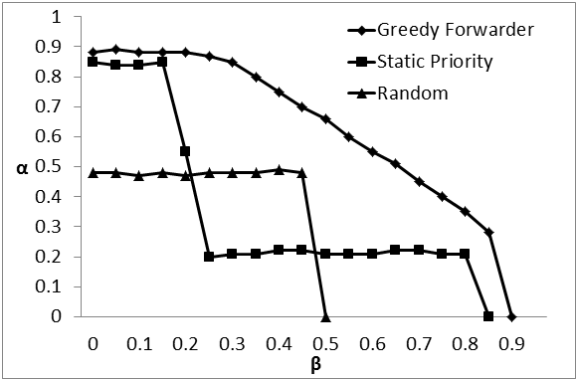}
\caption{Simulation results for the full-duplex system.} \label{fig:simulation:full-duplex}
\end{figure}

We also investigate the influence on the Greedy Forwarder when sensors only have delayed information on debts of flows. We assume that all sensors other than the sink update their information on debts every $\lambda$ intervals, and we call $\lambda$ the \emph{update period}. When a sensor $n$ updates its information, it notifies its children in the routing tree the information on debts that it currently has, and receives an updated information from its parent. Thus, for a sensor that is $g$-hop away from the sink, the information on debts that it has may be $\lambda\times g$ intervals old. We consider three scenarios: one where all sensors have knowledge of the current debts of flows, one where the update period is 100 intervals, i.e., 2 seconds, and one where the update period is 200 intervals. Simulation results are presented in Fig. \ref{fig:simulation:delay}. It can be shown that even when sensors update their information on debts as infrequent as once every four seconds, the performance of the Greedy Forwarder is still close to optimal.

\begin{figure}[t]
\includegraphics[height = 3.0in, angle = 90]{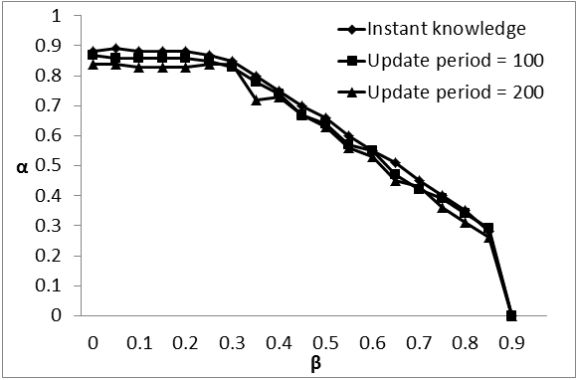}
\caption{Simulation results for systems where sensors only have delayed information on debts.} \label{fig:simulation:delay}
\end{figure}

Next we consider the half-duplex system. We consider a policy that, in each time slot, randomly selects a maximal set of sensors who can transmit simultaneously among those that hold some packets. Each selected sensor then randomly selects a packet to transmit. We call this policy the Random policy. We also consider a policy that first sorts all undelivered packets in descending order of the timely-throughput requirements of their associated flows. The policy then greedily selects a maximal subset of packets so that they can be transmitted simultaneously. This policy is called the Static Priority policy. Finally, we consider the Closest Sensor First policy.

The simulation results of the half-duplex system is shown in Fig. \ref{fig:simulation:half-duplex}. The Closest Sensor First policy achieves the largest timely-throughput region. A somewhat surprising result is that both the Random policy and the Static Priority policy have very poor performance. The Rand policy fails to fulfill the system even when we set $(\alpha, \beta)=(0,0.05)$, and hence its timely-throughput region does not appear in Fig. \ref{fig:simulation:half-duplex}. The reason for this behavior is because there are interference constraints for half-duplex systems, which limit the number of sensors that can transmit simultaneously. Thus, it is important to take the network topologies into account in order to deliver packets on time.

\begin{figure}[t]
\includegraphics[height = 3.0in, angle = 90]{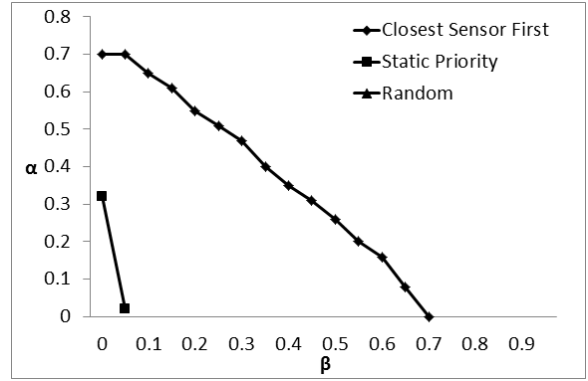}
\caption{Simulation results for the half-duplex system.} \label{fig:simulation:half-duplex}
\end{figure}

Finally, we consider a half-duplex path-topology system with six sensors and six flows, as depicted in Fig. \ref{fig:simulation:line_example}. All six flows are generated by sensor 5. Three of the flows require a timely-throughput of $\alpha$, while the other three flows require a timely-throughput of $\beta$. The simulation results are shown in Fig. \ref{fig:simulation:line}. The Closest Sensor First policy achieves the largest timely-throughput region. Moreover, both the Static Priority policy and the Random policy fail to fulfill the system even when we set $(\alpha, \beta) = (0,0.05)$.

\begin{figure}[t]
\includegraphics[height = 3.0in, angle = 90]{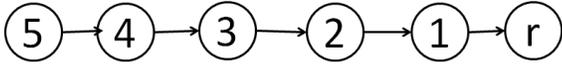}
\caption{The network topology for the half-duplex path-topology system.} \label{fig:simulation:line_example}
\end{figure}

\begin{figure}[t]
\includegraphics[height = 3.0in, angle = 90]{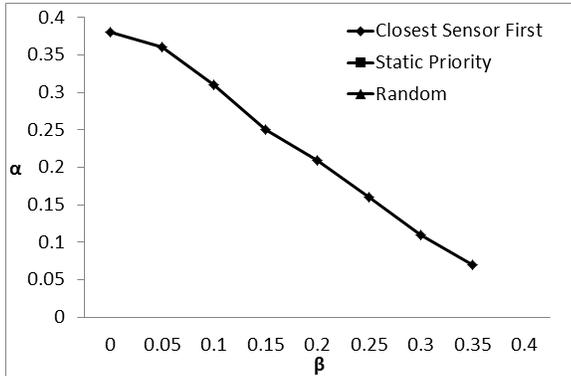}
\caption{Simulation results for the half-duplex path-topology system.} \label{fig:simulation:line}
\end{figure}

\section{Concluding Remarks}    \label{section:conclusion}

We have investigated the problem of providing hard per-packet delay guarantees for multi-hop wireless sensor networks. We have proposed an analytical model that jointly considers the hard delay guarantees of packets, the multi-hop routing tree of the network, the timely-throughput requirements of flows, and the unreliable nature of wireless transmissions. The model can be applied for both full-duplex systems and half-duplex systems. We have then introduced a framework for designing feasibility-optimal scheduling policies for different types of systems. Based on this framework, we have proposed a distributed scheduling policy for full-duplex systems and proved that this policy is feasibility-optimal. We have also proposed a heuristic for half-duplex systems. We have proved that this heuristic is feasibility-optimal for all path-topology systems. Simulation results have suggested that our proposed policies achieve much better performance than other policies.

\def\baselinestretch{1.0}
\bibliographystyle{ieeetr}
\bibliography{reference}

\begin{thebibliography}{10}

\bibitem{IS07}
I.~Stoianov, L.~Nachman, S.~Madden, and T.~Tokmouline, ``Pipenet: A wireless
  sensor network for pipeline monitoring,'' in {\em Proc. of ACM IPSN},
  pp.~264--273, 2007.

\bibitem{SO06}
S.~Oh, P.~Chen, M.~Manzo, and S.~Sastry, ``Instrumenting wireless sensor
  networks for real-time surveillance,'' in {\em Proc. of IEEE ICRA},
  pp.~3128--3133, 2006.

\bibitem{XZ11}
X.~Zhu, S.~Han, P.-C. Huang, A.~K. Mok, and D.~Chen, ``Mbstar : A real-time
  communication protocol for wireless body area networks,'' in {\em Proc. of
  ECRTS}, pp.~57--66, 2011.

\bibitem{PJ08}
P.~Jayachandran and T.~Abdelzaher, ``Transforming distributed acyclic systems
  into equivalent uniprocessors under preemptive and non-preemptive
  scheduling,'' in {\em Proc. of ECRTS}, pp.~233--242, 2008.

\bibitem{SH11}
S.~Hong, T.~Chantem, and X.~S. Hu, ``Meeting end-to-end deadlines through
  distributed local deadline assignments,'' in {\em Proc. of IEEE RTSS},
  pp.~183--192, 2011.

\bibitem{RL11}
R.~Li and A.~Eryilmaz, ``Scheduling for end-to-end deadline-constrained traffic
  with reliability requirements in multi-hop networks,'' in {\em Proc. of IEEE
  INFOCOM}, pp.~3065--3073, 2011.

\bibitem{VR10}
V.~Rodoplu, S.~Vadvalkar, A.~A. Gohari, and J.~J. Shynk, ``Empirical modeling
  and estimation of end-to-end voip delay over mobile multi-hop wireless
  networks,'' in {\em Proc. of IEEE Globecom}, 2010.

\bibitem{HL09}
H.~Li, Y.~Cheng, C.~Zhou, and W.~Zhuang, ``Minimizing end-to-end delay: A novel
  routing metric for multi-radio wireless mesh networks,'' in {\em Proc. of
  IEEE INFOCOM}, pp.~46--54, 2009.

\bibitem{PJ10}
P.~Jayachandran and M.~Andrews, ``Minimizing end-to-end delay in wireless
  networks using a coordinated edf schedule,'' in {\em Proc. of IEEE INFOCOM},
  2010.

\bibitem{HL11}
H.~Li, X.~Liu, W.~He, J.~Li, and W.~Dou, ``End-to-end delay analysis in
  wireless network coding: A network calculus-based approach,'' in {\em Proc.
  of IEEE ICDCS}, pp.~47--56, 2011.

\bibitem{JL09}
J.~Li, Z.~Li, and P.~Mohapatra, ``Adaptive per hop differentiation for
  end-to-end delay assurance in multihop wireless networks,'' {\em Ad Hoc
  Networks}, vol.~7, Aug. 2009.

\bibitem{BJ09}
B.~Jiang, B.~Ravindran, and H.~Cho, ``On real-time capacity of event-driven
  data-gathering sensor networks,'' in {\em Proc. of ACM MobiQuitous}, 2009.

\bibitem{XW10}
X.~Wang, X.~Wang, G.~Xing, and Y.~Yao, ``Dynamic duty cycle control for
  end-to-end delay guarantees in wireless sensor networks,'' in {\em Proc. of
  IEEE IWQoS}, 2010.

\bibitem{OC11}
O.~Chipara, C.~Wu, C.~Lu, and W.~Griswold, ``Interference-aware real-time flow
  scheduling for wireless sensor networks,'' in {\em Proc. of ECRTS},
  pp.~67--77, 2011.

\bibitem{YW09}
Y.~Wang, M.~C. Vuran, and S.~Goddard, ``Cross-layer analysis of the end-to-end
  delay distribution in wireless sensor networks,'' in {\em Proc. of IEEE
  RTSS}, pp.~138--147, 2009.

\bibitem{QW11}
Q.~Wang, P.~Fan, D.~O. Wu, and K.~B. Letaief, ``End-to-end delay constrained
  routing and scheduling for wireless sensor networks,'' in {\em Proc. of IEEE
  ICC}, 2011.

\bibitem{HL05}
H.~Li, P.~Shenoy, and K.~Ramamritham, ``Scheduling messages with deadlines in
  multi-hop real-time sensor networks,'' in {\em Proc. of IEEE RTAS}, pp.~415
  -- 425, 2005.

\bibitem{IH09}
I.-H. Hou, V.~Borkar, and P.~R. Kumar, ``A theory of {QoS} for wireless,'' in
  {\em Proc. of INFOCOM}, pp.~486--494, 2009.

\bibitem{AK04}
J.~Al-Karaki and A.~Kamal, ``Routing techniques in wireless sensor networks: a
  survey,'' {\em IEEE Wireless Communications}, vol.~11, pp.~6--28, Dec. 2004.

\bibitem{IH10}
I.-H. Hou and P.~R. Kumar, ``Scheduling heterogeneous real-time traffic over
  fading wireless channels,'' in {\em Proc. of IEEE INFOCOM}, 2010.

\end{thebibliography}
\end{document}